\documentclass[]{llncs}
\usepackage{csquotes}
\usepackage{amsmath}
\usepackage{centernot}
\usepackage{ltl}
\usepackage{tikz}
\usepackage{pgfplots}
\usepackage{xspace}
\usepackage{mathtools}
\usepackage{arydshln}
\usepackage{multirow}
\usepackage{subfig}

\definecolor{red}{RGB}{102,0,0}
\definecolor{green}{RGB}{102,204,0}
\definecolor{yellow}{RGB}{0,0,0}

\newcommand{\U}{\mbox{$\, {\sf U}\,$}}
\newcommand{\R}{\mbox{$\, {\sf R}\,$}}
\newcommand{\F}{\mbox{$\sf F$}}
\newcommand{\G}{\mbox{$\sf G$}}
\newcommand{\X}{\mbox{$\sf X$}}
\newcommand{\WU}{\mbox{$\, {\sf W}\,$}}

\renewcommand{\LTLsquare}{\G}
\renewcommand{\LTLdiamond}{\F}

\newcommand{\V}{\mathcal{V}}

\newcommand{\AP}{\mathit{AP}}

\usepackage{wrapfig}

\newcommand{\tool}{\textsc{MGHyper}\xspace}

\begin{document}
	\title{\tool: Checking Satisfiability of HyperLTL Formulas Beyond the $\exists^*\forall^*$ Fragment\thanks{This work was partially supported by the German Research Foundation (DFG) in the Collaborative Research Center 1223 and by the European Research Council (ERC) Grant OSARES (No. 683300).}}
	\author{Bernd Finkbeiner \and Christopher Hahn \and Tobias Hans}
	\institute{Reactive Systems Group\\Saarland University\\\email{lastname@react.uni-saarland.de}}
	\maketitle
\vspace{-5ex}	
	\begin{abstract}
          Hyperproperties are properties that refer to multiple
          computation traces. This includes many information-flow
          security policies, such as observational determinism,
          (generalized) noninterference, and noninference, and other
          system properties like symmetry or Hamming distances between
          in error-resistant codes.  We introduce \tool, a tool for
          automatic satisfiability checking and model generation for
          hyperproperties expressed in HyperLTL. Unlike previous
          satisfiability checkers, \tool is not limited to the
          decidable $\exists^*\forall^*$ fragment of HyperLTL, but
          provides a semi-decisionprocedure for the full logic.  An
          important application of \tool is to automatically check
          equivalences between different hyperproperties (and
          different formalizations of the same hyperproperty) and to
          build counterexamples that disprove a certain claimed
          implication. We describe the semi-decisionprocedure
          implemented in \tool and report on experimental results
          obtained both with typical hyperproperties from the
          literature and with randomly generated HyperLTL formulas.
%
\vspace{-2ex}
	\end{abstract}
	
	\section{Introduction}
	\vspace{-1ex}
	\sloppy
	\label{intro}
	HyperLTL~\cite{clarkson2014temporal} extends linear-time temporal logic (LTL)~\cite{pnueli1977temporal} with explicit quantification over traces. This makes it possible to express hyperproperties~\cite{clarkson2010hyperproperties} like noninterference~\cite{DBLP:conf/sp/GoguenM82a} or symmetry~\cite{DBLP:conf/cav/FinkbeinerRS15}, which refer to multiple traces at the same time. Such properties are not expressible in LTL, or even in the branching-time temporal logics CTL~\cite{DBLP:conf/lop/ClarkeE81} and CTL$^*$~\cite{DBLP:journals/jacm/EmersonH86}.
	For example, \emph{noninference}~\cite{DBLP:journals/tse/McLean96} is a variant of noninterference stating that, for all system traces, the low-observable behavior must not change when all high inputs are replaced by a dummy input. The following HyperLTL formula expresses this policy:
	$\forall \pi.\;\exists \pi'.\;(\G \lambda_{\pi'} ) \wedge \pi =_{L, out} {\pi'}.$
	HyperLTL is supported by model checking~\cite{DBLP:conf/cav/FinkbeinerRS15} and runtime monitoring tools~\cite{DBLP:conf/rv/FinkbeinerHST17,DBLP:conf/tacas/FinkbeinerHST18}. There is also a decision procedure, EAHyper~\cite{finkbeiner2017eahyper}, which checks the satisfiability of a given formula from the $\exists^*\forall^*$ fragment of HyperLTL. EAHyper is based on a reduction from HyperLTL satisfiability to LTL satisfiability~\cite{DBLP:conf/concur/FinkbeinerH16}. A major application of EAHyper is to check equivalences between alternation-free HyperLTL formulas, i.e., formulas that either contain only universal quantifiers or only existential quantifers. Such equivalences can be expressed in the $\exists^*\forall^*$ fragment. It is impossible, however, to handle formulas that contain a $\forall \exists$ quantifier alternation, as, for example, in noninference. This is unfortunate, because such a quantifier alternation is often needed, in particular to account for nondeterminism. A popular example is \emph{generalized noninterference}~\cite{DBLP:conf/sp/McCullough88}:
$\forall\pi.\forall\pi'.\exists\pi''.~\pi\!=_{H,\mathsf{in}}\!\pi'' \wedge \pi'\!=_{L, out}\!\pi''$.
The formula expresses that for every possible high-security input (seen on some trace $\pi$) and every possible low-security observations (seen on some trace $\pi'$) there exists a nondeterministic execution $\pi''$ where the high-security input and the low-security observations happen together. Hence, the observer cannot conclude, after making the low-security observations, that any specific high-security input actually occurred. Other properties that need a $\forall \exists$ quantifier alternation include restrictiveness~\cite{DBLP:journals/tse/McCullough90}, separability~\cite{DBLP:conf/sp/McLean94}, and forward correctability~\cite{DBLP:journals/jcs/Millen95}.
For formulas outside the $\exists^*\forall^*$ fragment, it is no longer possible to reduce the HyperLTL satisfiability problem to the LTL satsifiability problem: the HyperLTL satisfiability problem is, in fact, undecidable~\cite{DBLP:conf/concur/FinkbeinerH16}. In this paper, we present the first semi-decisionprocedure for full HyperLTL. Our approach is based on a reduction to quantified boolean formulas (QBF)~\cite{DBLP:books/fm/GareyJ79} and has been implemented in the \tool tool. \tool can be used to analyze and develop hyperproperties and, especially, generate models that disprove equivalences or implications between
different hyperproperties or different formalizations of the same hyperproperty. For example, comparing noninference to generalized noninterference, \tool instantly demonstrates that the two properties are \emph{not} equivalent.
	\vspace{-2ex}
\section{A Semi-Decision Procedure for HyperLTL-SAT}
\vspace{-1ex}
A \emph{hyperproperty} is a \emph{set of sets of traces}.
Hyperproperties can be expressed in HyperLTL, which generalizes LTL with explicit trace quantification:
\vspace{-1ex}
\begin{align*}
\psi~&::=~\exists\pi.\;\psi~~|~~\forall\pi.\;\psi~~|~~\varphi \\
\varphi~&::=~a_\pi~~|~\neg \varphi~~|~~\varphi~\vee~\varphi~~|~~\X\;\varphi~~|~~\varphi\, \U \varphi~~|~~true
\end{align*}
where $Q$ is an existential or universal quantifier, $a \in \AP$ is an atomic proposition and $\pi \in \V$ is a trace variable of an infinitely supply $\V$.
Logical connectives and the temporal operators  $\F$, $\G$, $\WU$ and $\R$ are defined as in LTL.
The semantics of HyperLTL is defined as follows.
\vspace{-1ex}
\begin{align*}
&\Pi \models_T~\exists\pi. \psi &&\text{iff}\hspace{5ex} \text{there exists}~t \in T~:~ \Pi[\pi \mapsto t] \models_T \psi \\
&\Pi \models_T~\forall\pi. \psi &&\text{iff}\hspace{5ex} \text{for all}~t \in T~:~ \Pi[\pi \mapsto t] \models_T \psi \\
&\Pi \models_T~a_\pi &&\text{iff}\hspace{5ex} a \in \Pi(\pi)[0] \\
&\Pi \models_T~\neg \psi &&\text{iff}\hspace{5ex} \Pi \not \models_T \psi \\
&\Pi \models_T~\psi_1 \vee \psi_2 &&\text{iff}\hspace{5ex} \Pi \models_T \psi_1~\text{or}~\Pi \models_T \psi_2 \\
&\Pi \models_T~\X\psi &&\text{iff}\hspace{5ex} \Pi[1,\infty] \models_T \psi \\
&\Pi \models_T~\psi_1 \U \psi_2 &&\text{iff}\hspace{5ex} \text{there exists}~i \geq 0 : \Pi[i,\infty] \models_T \psi_2 \\
& &&\hspace{7.2ex} \text{and for all}~0 \leq j < i~\text{we have}~\Pi[j,\infty] \models_T \psi_1
\end{align*}
,where $\Pi:\V \mapsto TR$ is the trace assignment function, which maps trace variables to traces, denoted by $\Pi[\pi\mapsto t]$. 
The suffixes of all traces $\pi$ starting at step $i$ is denoted by $\Pi [ i, \infty ] $.
\emph{HyperLTL-SAT} is the problem to decide, if a non-empty trace set $T$ exists, such that $\{\} \models_T \psi$.

\tool takes an \emph{arbitrary} HyperLTL formula of the following form as input: $Q_0 \vec{\pi_0} \ldots Q_n \vec{\pi_n} . \varphi $, where $Q_i \in \{ \exists, \forall\}$ and $\vec{\pi_0} \dots \vec{\pi_n}$ are vectors over $\mathcal{V}$. \tool evaluates to ``sat'' if  and only if the formula is  satisfiable. Basically, \tool checks whether or not there exists a trace set of size $m$ that satisfies the HyperLTL formula under consideration. 
The procedure starts with trace sets of size $1$, and increment $m$ until a witness is found, leading to the following theorem.
\vspace{-0.5ex}
\begin{theorem}\label{lemma}
	HyperLTL-SAT is RE-complete.
\end{theorem}
\begin{proof}
	We prove membership by constructing a QBF formula $\varphi_{QBF}^m$, which is satisfiable if the given HyperLTL formula $\varphi$ is satisfiable by a trace set of size $m$.
	The basic idea of the encoding of a HyperLTL formula to a QBF formula is threefold: 1) we construct a quantifier prefix that resembles the quantifier structure in the given HyperLTL formula, 2) we construct a premise that links trace variables to actual traces, and 3) we unroll the LTL suffix into a SAT-encoding. The third step follows the unrolling presented in~\cite{biere2003bounded} and will, due to space reasons, not be discussed. We refer to the maximum trace-unrolling bound as $k$ (not to confuse with $m$), which is exponential in the size of the LTL suffix.
	\paragraph{1) Prefix.} Let S be a set and $k$ a natural number, we define ${\mathit{Traces}^k_S}$ as ${\{ a_s^i \mid 0 \leq i < k ,\forall s \in S,\forall a \in \AP  \}}$, which we use as the trace representation inside the QBF encoding. Let $\varphi:=Q_0 \vec{\pi_0} \ldots Q_n \vec{\pi_n} . \psi $ be a HyperLTL formula.
	The quantifier prefix of the resulting QBF introduces existential quantifiers, representing the trace set $T$ of size $m$ (the witness for satisfiability). The trace variables are quantified according to their quantifier in the HyperLTL formula:
	\vspace{-1ex}
	\[ \mathit{Prefix}(\varphi) = \exists \mathit{Traces}^k_{T}.Q_0 \mathit{Traces}^k_{\vec{\pi_0}}.Q_1 \mathit{Traces}^k_{\vec{\pi_1}}.\ldots . Q_n \mathit{Traces}^k_{\vec{\pi_n}}. \]
	
	\vspace{-3ex}
	
	\paragraph{2) Linking.}	
	We construct a premise to link trace variables to the trace witnesses in $\mathit{Traces}^k_T$.
	For every quantifier $Q_i$ a subpremise $P_{Q_i}$ is constructed first, which represents the mapping of all trace variables in $\vec{\pi_i}$.
Mapping each trace variable to traces reassembles the trace assignment function from the HyperLTL semantics and is encoded by ensuring that the boolean variables with the same atomic proposition in the same step share the same truth value.
%
	\begin{equation}
	P_{Q_i} := \Bigg[  \bigwedge_{\pi \in \vec{\pi_i}} \bigvee_{t_i \in T} 
	\Bigg[  \bigwedge_{\substack{ (a^j_{t_i} , a^j_{\pi}) \in \\ \mathit{Traces}^k_{t_i} \times \mathit{Traces}^k_{\pi} } }a^j_{t_i} \leftrightarrow a^j_{\pi}      \Bigg]  \Bigg] 
	\end{equation}
	The linking mechanism is a combination of the subpremises. The boolean connective between the different supremises depends on the corresponding quantifier:
	\begin{equation}
	\mathit{Linking}(\varphi) := P_{Q_0}^k \circ_{Q_0} P_{Q_1}^k \circ_{Q_1} \ldots P_{Q_{n-1}}^k \circ_{Q_{n-1}} P_{Q_n}^k \circ_{Q_{n}},
	\end{equation}
	where $ \circ_{Q_i} $ equals $\rightarrow $, if $Q_i = \forall $, and $ \circ_{Q_i}$ equals $\land$ if $Q_i = \exists$.
	Together with \emph{3)~the~unrolling} of the LTL suffix~\cite{biere2003bounded}, the  constructed $\varphi_{QBF}^m$  a QBF formula is satisfiable if the HyperLTL formula $\varphi$ is satisfiable. 
	Hardness follows from a reduction from Post's Correspondence Problem~\cite{DBLP:conf/concur/FinkbeinerH16}.\qed
\end{proof}

\begin{example}
	Consider the HyperLTL formula $\varphi:= \forall \pi_0 \exists \pi_1 \exists \pi_2 . (a_{\pi_0} \land  ( a_{\pi_1} \rightarrow  \neg b_{\pi_1} \land  a_{\pi_2} \rightarrow   b_{\pi_2}) )$. 
	Note that, for the sake of simplicity, the example LTL formula does not contain temporal operators. 
	In the first iteration, \tool tries to guess a trace set $T$ of size $1$ and will not find a satisfying assignment for the constructed QBF formula.
	In the second iteration, though, \tool constructs the following QBF formula, with $T_2= \{ \{ a^0_{t_0} , b^0_{t_0} \}, \{ a^0_{t_1} , b^0_{t_1} \} \}$.
	\begin{align*}
	&\exists \mathit{Traces}^0_{T_2} \forall \mathit{Traces}^0_{\pi_0} \exists \mathit{Traces}^0_{\pi_1} \exists \mathit{Traces}^0_{\pi_2}. 
	\left[ \lor
	\begin{tabular}{c}
	$( a_{\pi_0}^{0} \leftrightarrow  a_{t_0}^{0} \land  b_{\pi_0}^{0} \leftrightarrow  b_{t_0}^{0} ) $ \\
	$(  a_{\pi_0}^{0} \leftrightarrow  a_{t_1}^{0} \land  b_{\pi_0}^{0} \leftrightarrow  b_{t_1}^{0}) $ \\
	\end{tabular}
	\right] \\
	&\rightarrow 
	\left(  \left[ \land
	\begin{tabular} {r c r}
	$ \Big( \lor$ &
	\begin{tabular}{c}
	$[ a_{\pi_1}^{0} \leftrightarrow  a_{t_0}^{0} \land  b_{\pi_1}^{0} \leftrightarrow  b_{t_0}^{0} ] $ \\
	$[ a_{\pi_1}^{0} \leftrightarrow  a_{t_1}^{0} \land  b_{\pi_1}^{0} \leftrightarrow  b_{t_1}^{0} ] $ \\
	\end{tabular} \Big) \\
	$ \Big( \lor$ &
	\begin{tabular}{c}
	$[ a_{\pi_2}^{0} \leftrightarrow  a_{t_0}^{0} \land  b_{\pi_2}^{0} \leftrightarrow  b_{t_0}^{0} ] $ \\
	$[ a_{\pi_2}^{0} \leftrightarrow  a_{t_1}^{0} \land  b_{\pi_2}^{0} \leftrightarrow  b_{t_1}^{0} ] $ \\
	\end{tabular} \Big) \\
	\end{tabular}
	\right] 
	\wedge (a^0_{\pi_0} \land  ( a^0_{\pi_1} \rightarrow  \neg b^0_{\pi_1} \land  a^0_{\pi_2} \rightarrow   b^0_{\pi_2}) )  \right) \\
	\end{align*}
	This QBF formula is satisfied by the assignment $ \mathcal{A} =  \{ a^0_{t_0}, b^0_{t_0}, a^0_{t_1} , \neg b^0_{t_1} \}$ for the existentially quantified variables, which represent the traces (of length one in this example).
	There are four possible assignment for the universally quantified boolean variables. 
	For  $ \{a^0_{\pi_0} , b^0_{\pi_0}\} $ or $ \{a^0_{\pi_0} ,  \neg b^0_{\pi_0} \}$ we can map the existentially quantified traces variables $\pi_1\mapsto t_1$ and $ \pi_2 \mapsto t_0$,
	which add $ \{a_{\pi_1}^{0} , \neg b_{\pi_1}^{0}, a_{\pi_2}^{0}, b_{\pi_2}^{0}\}$ to $\mathcal{A}$, such that $\mathcal{A}$ satisfies the formula.
	In the other two cases,$\{\neg a^0_{\pi_0}, \neg b^0_{\pi_0} \}$ or $\{ \neg a^0_{\pi_0} ,  b^0_{\pi_0} \}$, we cannot map to $\neg a^0_{\pi_0}$, which leads to a false evaluation of 
	$P_{Q_0}^0$ and therefore to a true evaluation of the formula. 
	From $\mathcal{A}$ we can now follow that $\{ \{a,b\}^{\omega} , \{ a\}^{\omega} \}$  is a model that satisfies $\varphi$.
\end{example}


\section{Experimental Results}
\tool is implemented in Ocaml and supports UNIX-based operation systems.
We tested our tool against different benchmarks on a virtual machine
running an Ubuntu (64-Bit) 14.04LTS installation on an Intel Core i7-4710MQ with 2,50GH on 4 kernels and 8GB RAM.

\vspace{-2ex}

\noindent
\paragraph{Counter Examples for Implication of Security Polices.}
A main application of \tool is to check if two arbitrary HyperLTL formulas do not imply each other: we check if the negation of the implication is satisfiable.
The first benchmark checks the implication of 
\emph{General Noninterference}~\cite{clarkson2010hyperproperties} ((GNI):$\forall \pi_1 .\forall \pi_2. \exists \pi_3 \G( I^{high}_{\pi_1} = I^{high}_{\pi_3} ) \land \G(O^{low}_{\pi_2} = O^{low}_{\pi_3} )$),
\emph{Noninerference} ((NI): 	$\forall \pi_1.\;\exists \pi_2.\;(\G \lambda_{\pi_{2}} ) \wedge \G (O_{\pi_1} = O_{\pi_2})$)
and several formalizations of \emph{Observational Determinsim}~\cite{DBLP:journals/jcs/McLean92,DBLP:conf/sp/Roscoe95,DBLP:conf/csfw/ZdancewicM03}:
(OD): $\forall \pi_1 .\forall \pi_2.( I^{low}_{\pi_1} = I^{low}_{\pi_2} ) \rightarrow \G(O^{low}_{\pi_1} = O^{low}_{\pi_2} ) $,
(GOD): $\forall \pi_1 .\forall \pi_2. \G ( I^{low}_{\pi_1} = I^{low}_{\pi_2} ) \rightarrow \G (O^{low}_{\pi_1} = O^{low}_{\pi_2} )$, and
(WOD):$ \forall \pi_1 .\forall \pi_2.( I^{low}_{\pi_1} = I^{low}_{\pi_2} ) \WU (O^{low}_{\pi_1} \neq O^{low}_{\pi_2} )$.
\tool shows that none of the formalizations of observational determinism does imply general non-interference or noninfernce.
Furthermore, it shows that generalized non-interference does not imply noninference. Every check was done in under $0.05$ seconds.
\begin{figure}[t]
	\centering
	\scalebox{0.8}{
		\begin{tikzpicture} 
		\begin{axis}[ height=4.5cm, width=\textwidth, grid=major,legend style={at={(1,0.5)},anchor=west}]
		\addplot[red,mark=square*] coordinates {(1,0.350768)(2,1.16265)(3,0.72175)(4,1.53487)(5,0.336617)(6,0.0233755)(7,0.0754869)(8,1.08922)(9,0.140393)(10,2.25573)(11,0.145131)(12,1.76749)(13,0.111929)(14,1.42211)(15,0.217322)(16,0.0995075)(17,0.177371)(18,2.07337)(19,0.133035)(20,0.170281)(21,0.169829)(22,0.170496)(23,1.36351)(24,0.210785)(25,0.398731)(26,0.378725)(27,0.264689)(28,0.290237)(29,0.517447)(30,0.356709)(31,2.06768)(32,0.357807)(33,0.359984)(34,0.39064)(35,0.478758)(36,0.432273)(37,0.511511)(38,0.490323)(39,0.730603)(40,0.559758)(41,1.11724)(42,0.606257)(43,1.83409)(44,0.671538)(45,0.758744)(46,1.02239)(47,0.70656)(48,0.790652)(49,1.21716)};
		\addlegendentry{existential60} 
		\addplot[orange,mark=square*]  coordinates {(1,0.00596244)(2,0.00903679)(3,0.012433)(4,0.0168263)(5,0.0212998)(6,0.0246188)(7,0.0304422)(8,0.0370211)(9,0.0429337)(10,0.0523534)(11,0.0668932)(12,0.0694483)(13,0.0724185)(14,0.11926)(15,0.0924521)(16,0.140381)(17,0.115462)(18,0.119172)(19,0.142097)(20,0.21753)(21,0.170035)(22,0.275456)(23,0.189852)(24,0.228936)(25,0.241302)(26,0.266707)(27,0.257497)(28,0.290084)(29,0.303961)(30,0.452011)(31,0.348642)(32,0.508721)(33,0.386017)(34,0.385726)(35,0.427942)(36,0.491389)(37,0.49606)(38,0.496025)(39,0.522784)(40,0.539046)(41,0.571756)(42,0.664848)(43,0.749031)(44,0.687544)(45,0.676989)(46,0.683653)(47,1.161)(48,1.1118)(49,0.835144)}; 
		\addlegendentry{universal60} 
		\addplot[blue,mark=*] coordinates {(2,1.06787)(3,1.00044)(4,2.41057)(5,1.38481)(6,1.85294)(7,0.117811)(8,0.389956)(9,2.58598)(10,1.98649)(11,1.24449)(12,2.68977)(13,0.626694)(14,0.14336)(15,0.304192)(16,0.238961)(17,0.72666)(18,0.329054)(19,0.300775)(20,0.556673)(21,2.28714)(22,0.631316)(23,0.268599)(24,0.575532)(25,0.71572)(26,0.582621)(27,2.3372)(28,0.606393)(29,0.786538)(30,0.812918)(31,0.818921)(32,0.700189)(33,1.03475)(34,0.81919)(35,0.685939)(36,2.12413)(37,0.942957)(38,1.10373)(39,1.67235)(40,1.19079)(41,1.61702)(42,1.60099)(43,1.09931)(44,1.43455)(45,1.73156)(46,1.36217)(47,1.15813)(48,3.15988)(49,2.249)}; 
		\addlegendentry{existential120} 
		\addplot[green,mark=*] coordinates {(1,0.0764965)(2,0.0125866)(3,0.0580271)(4,0.653877)(5,0.055682)(6,0.0410031)(7,0.255329)(8,0.0740704)(9,0.145155)(10,0.0760372)(11,2.02249)(12,0.0948335)(13,0.132947)(14,0.122116)(15,0.303246)(16,0.201181)(17,0.432695)(18,0.177761)(19,0.474722)(20,0.278293)(21,0.540365)(22,0.434694)(23,0.747256)(24,0.342119)(25,0.746703)(26,0.374791)(27,0.519198)(28,0.431083)(29,0.830275)(30,0.722988)(31,0.748451)(32,0.574963)(33,1.27606)(34,0.73839)(35,1.06681)(36,0.875465)(37,1.92106)(38,0.799709)(39,1.17212)(40,1.32736)(41,1.11198)(42,1.3807)(43,1.32587)(44,1.21488)(45,2.57465)(46,1.10059)(47,1.43692)(48,1.25609)(49,3.05459)}; 
		\addlegendentry{universal120} 
		\end{axis} 
		\end{tikzpicture}
	}
	\caption{Runtime for random HyperLTL formulas of size 60 and 120. A formula consists of up to $50$ trace variables starting universally or existentially quantified.}
	\label{rand_alt_timeout120}
\end{figure}
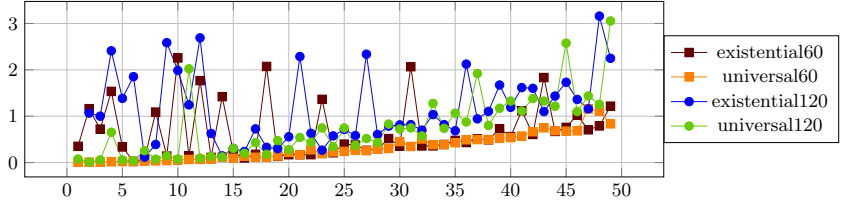


\noindent
\paragraph{Random Formulas.}
We tested \tool against different benchmarks of random formulas.
\emph{Quantfier Alternation:}
We created HyperLTL formulas with up to $49$ quantifier alternations and $15$ atomic propostion  using randltl~\cite{DBLP:conf/atva/Duret-Lutz13}.
For each number of alternations we tested $100$ formulas of size $60$ and $120$, where 50 start with $\exists$ and 50 with $\forall$. The size is the size argument provided for randltl~\cite{DBLP:conf/atva/Duret-Lutz13}. The runtimes are shown in Figure \ref{rand_alt_timeout120}.
\emph{Quantifier Ordering:}
\emph{$\exists^{n}\forall^{m}$ Formulas:} For the sake of comparing \tool with EAHyper, we tested \tool on the largest decidable fragment of HyperLTL, which is the $\exists^*\forall^*$-fragment. We scaled in the number of existential and universal quantifiers, showing that \tool is able to solve formulas with up to $10$ existential and $10$ universal quantifier. In comparison, EAHyper, implementing the first decision procedure for HyperLTL-SAT, already runs out of memory after $5$ existential and $5$ universal quantifiers.

\begin{table}[t]
	\centering
	\caption{Random $\exists^n\forall^m$ formulas: instances solved (sol) in $120$ $s$ and average time (avgt) in $s$ for $100$ random formulas of size $60$ with $15$ atomic propositions.}
	\label{results}
        \smallskip
	\scalebox{0.8}{
	\begin{tabular}{|cc:c|c:c|c:c|c:c|c:c|c:c|c:c|c:c|c:c|c:c|}
	&avgt&sol&avgt&sol&avgt&sol&avgt&sol&avgt&sol&avgt&sol&avgt&sol&avgt&sol&avgt&sol&avgt&sol\\ 
\hline 
&\multicolumn{2}{c|}{$\exists^{1}\forall^{1}$}&\multicolumn{2}{c|}{$\exists^{1}\forall^{2}$}&\multicolumn{2}{c|}{$\exists^{1}\forall^{3}$}&\multicolumn{2}{c|}{$\exists^{1}\forall^{4}$}&\multicolumn{2}{c|}{$\exists^{1}\forall^{5}$}&\multicolumn{2}{c|}{$\exists^{1}\forall^{6}$}&\multicolumn{2}{c|}{$\exists^{1}\forall^{7}$}&\multicolumn{2}{c|}{$\exists^{1}\forall^{8}$}&\multicolumn{2}{c|}{$\exists^{1}\forall^{9}$}&\multicolumn{2}{c|}{$\exists^{1}\forall^{10}$}\\ 
& 0.472& 96& 0.599& 95& 1.371& 96& 1.537& 96& 1.322& 96& 1.587& 96& 0.529& 92& 2.823& 96& 3.166& 94& 2.303& 97\\ 
\hline 
&\multicolumn{2}{c|}{$\exists^{2}\forall^{1}$}&\multicolumn{2}{c|}{$\exists^{2}\forall^{2}$}&\multicolumn{2}{c|}{$\exists^{2}\forall^{3}$}&\multicolumn{2}{c|}{$\exists^{2}\forall^{4}$}&\multicolumn{2}{c|}{$\exists^{2}\forall^{5}$}&\multicolumn{2}{c|}{$\exists^{2}\forall^{6}$}&\multicolumn{2}{c|}{$\exists^{2}\forall^{7}$}&\multicolumn{2}{c|}{$\exists^{2}\forall^{8}$}&\multicolumn{2}{c|}{$\exists^{2}\forall^{9}$}&\multicolumn{2}{c|}{$\exists^{2}\forall^{10}$}\\ 
& 6.772& 80& 2.743& 85& 2.698& 93& 6.319& 94& 4.233& 97& 4.107& 92& 3.162& 87& 2.906& 94& 4.847& 92& 2.131& 96\\ 
\hline 
&\multicolumn{2}{c|}{$\exists^{3}\forall^{1}$}&\multicolumn{2}{c|}{$\exists^{3}\forall^{2}$}&\multicolumn{2}{c|}{$\exists^{3}\forall^{3}$}&\multicolumn{2}{c|}{$\exists^{3}\forall^{4}$}&\multicolumn{2}{c|}{$\exists^{3}\forall^{5}$}&\multicolumn{2}{c|}{$\exists^{3}\forall^{6}$}&\multicolumn{2}{c|}{$\exists^{3}\forall^{7}$}&\multicolumn{2}{c|}{$\exists^{3}\forall^{8}$}&\multicolumn{2}{c|}{$\exists^{3}\forall^{9}$}&\multicolumn{2}{c|}{$\exists^{3}\forall^{10}$}\\ 
& 5.533& 79& 9.921& 81& 5.836& 82& 4.851& 88& 7.589& 82& 3.852& 82& 9.975& 82& 5.222& 79& 6.328& 77& 5.044& 83\\ 
\hline 
&\multicolumn{2}{c|}{$\exists^{4}\forall^{1}$}&\multicolumn{2}{c|}{$\exists^{4}\forall^{2}$}&\multicolumn{2}{c|}{$\exists^{4}\forall^{3}$}&\multicolumn{2}{c|}{$\exists^{4}\forall^{4}$}&\multicolumn{2}{c|}{$\exists^{4}\forall^{5}$}&\multicolumn{2}{c|}{$\exists^{4}\forall^{6}$}&\multicolumn{2}{c|}{$\exists^{4}\forall^{7}$}&\multicolumn{2}{c|}{$\exists^{4}\forall^{8}$}&\multicolumn{2}{c|}{$\exists^{4}\forall^{9}$}&\multicolumn{2}{c|}{$\exists^{4}\forall^{10}$}\\ 
& 10.316& 75& 8.669& 80& 5.631& 83& 3.722& 83& 5.843& 73& 5.62& 81& 10.074& 79& 6.955& 76& 8.037& 85& 5.938& 79\\ 
\hline 
&\multicolumn{2}{c|}{$\exists^{5}\forall^{1}$}&\multicolumn{2}{c|}{$\exists^{5}\forall^{2}$}&\multicolumn{2}{c|}{$\exists^{5}\forall^{3}$}&\multicolumn{2}{c|}{$\exists^{5}\forall^{4}$}&\multicolumn{2}{c|}{$\exists^{5}\forall^{5}$}&\multicolumn{2}{c|}{$\exists^{5}\forall^{6}$}&\multicolumn{2}{c|}{$\exists^{5}\forall^{7}$}&\multicolumn{2}{c|}{$\exists^{5}\forall^{8}$}&\multicolumn{2}{c|}{$\exists^{5}\forall^{9}$}&\multicolumn{2}{c|}{$\exists^{5}\forall^{10}$}\\ 
& 5.431& 71& 5.009& 80& 2.812& 69& 8.514& 81& 4.501& 76& 6.255& 83& 1.574& 76& 3.616& 76& 5.85& 79& 7.486& 80\\ 
\hline 
&\multicolumn{2}{c|}{$\exists^{6}\forall^{1}$}&\multicolumn{2}{c|}{$\exists^{6}\forall^{2}$}&\multicolumn{2}{c|}{$\exists^{6}\forall^{3}$}&\multicolumn{2}{c|}{$\exists^{6}\forall^{4}$}&\multicolumn{2}{c|}{$\exists^{6}\forall^{5}$}&\multicolumn{2}{c|}{$\exists^{6}\forall^{6}$}&\multicolumn{2}{c|}{$\exists^{6}\forall^{7}$}&\multicolumn{2}{c|}{$\exists^{6}\forall^{8}$}&\multicolumn{2}{c|}{$\exists^{6}\forall^{9}$}&\multicolumn{2}{c|}{$\exists^{6}\forall^{10}$}\\ 
& 3.53& 78& 4.378& 74& 3.503& 71& 3.057& 76& 4.354& 71& 4.513& 81& 3.492& 79& 4.836& 79& 6.289& 80& 6.0& 74\\ 
\hline 
&\multicolumn{2}{c|}{$\exists^{7}\forall^{1}$}&\multicolumn{2}{c|}{$\exists^{7}\forall^{2}$}&\multicolumn{2}{c|}{$\exists^{7}\forall^{3}$}&\multicolumn{2}{c|}{$\exists^{7}\forall^{4}$}&\multicolumn{2}{c|}{$\exists^{7}\forall^{5}$}&\multicolumn{2}{c|}{$\exists^{7}\forall^{6}$}&\multicolumn{2}{c|}{$\exists^{7}\forall^{7}$}&\multicolumn{2}{c|}{$\exists^{7}\forall^{8}$}&\multicolumn{2}{c|}{$\exists^{7}\forall^{9}$}&\multicolumn{2}{c|}{$\exists^{7}\forall^{10}$}\\ 
& 4.33& 74& 3.173& 70& 1.789& 72& 7.187& 69& 4.99& 78& 5.584& 74& 4.783& 77& 7.558& 75& 7.744& 74& 6.043& 78\\ 
\hline 
&\multicolumn{2}{c|}{$\exists^{8}\forall^{1}$}&\multicolumn{2}{c|}{$\exists^{8}\forall^{2}$}&\multicolumn{2}{c|}{$\exists^{8}\forall^{3}$}&\multicolumn{2}{c|}{$\exists^{8}\forall^{4}$}&\multicolumn{2}{c|}{$\exists^{8}\forall^{5}$}&\multicolumn{2}{c|}{$\exists^{8}\forall^{6}$}&\multicolumn{2}{c|}{$\exists^{8}\forall^{7}$}&\multicolumn{2}{c|}{$\exists^{8}\forall^{8}$}&\multicolumn{2}{c|}{$\exists^{8}\forall^{9}$}&\multicolumn{2}{c|}{$\exists^{8}\forall^{10}$}\\ 
& 5.681& 81& 4.617& 79& 6.803& 74& 5.563& 75& 6.219& 80& 5.999& 74& 3.013& 73& 2.041& 72& 6.146& 75& 3.997& 75\\ 
\hline 
&\multicolumn{2}{c|}{$\exists^{9}\forall^{1}$}&\multicolumn{2}{c|}{$\exists^{9}\forall^{2}$}&\multicolumn{2}{c|}{$\exists^{9}\forall^{3}$}&\multicolumn{2}{c|}{$\exists^{9}\forall^{4}$}&\multicolumn{2}{c|}{$\exists^{9}\forall^{5}$}&\multicolumn{2}{c|}{$\exists^{9}\forall^{6}$}&\multicolumn{2}{c|}{$\exists^{9}\forall^{7}$}&\multicolumn{2}{c|}{$\exists^{9}\forall^{8}$}&\multicolumn{2}{c|}{$\exists^{9}\forall^{9}$}&\multicolumn{2}{c|}{$\exists^{9}\forall^{10}$}\\ 
& 3.509& 77& 2.514& 72& 5.659& 68& 1.345& 78& 4.379& 72& 3.914& 73& 3.422& 71& 1.784& 66& 6.903& 72& 5.142& 77\\ 
\hline 
&\multicolumn{2}{c|}{$\exists^{10}\forall^{1}$}&\multicolumn{2}{c|}{$\exists^{10}\forall^{2}$}&\multicolumn{2}{c|}{$\exists^{10}\forall^{3}$}&\multicolumn{2}{c|}{$\exists^{10}\forall^{4}$}&\multicolumn{2}{c|}{$\exists^{10}\forall^{5}$}&\multicolumn{2}{c|}{$\exists^{10}\forall^{6}$}&\multicolumn{2}{c|}{$\exists^{10}\forall^{7}$}&\multicolumn{2}{c|}{$\exists^{10}\forall^{8}$}&\multicolumn{2}{c|}{$\exists^{10}\forall^{9}$}&\multicolumn{2}{c|}{$\exists^{10}\forall^{10}$}\\ 
& 3.986& 81& 4.553& 70& 5.777& 70& 4.791& 75& 8.284& 75& 1.534& 72& 4.338& 71& 4.18& 76& 5.512& 65& 4.529& 75\\ 
\hline 
	\end{tabular}
	}
\end{table}


\section{Conclusion}
\vspace{-1ex}
We have presented \tool, the first semi-decisionprocedure for checking the satisfiability of HyperLTL formulas beyond the decidable $\exists^*\forall^*$ fragment. An application of \tool is the analysis and development of hyperproperties and, especially, the generation of models that disprove equivalences or implications between different hyperproperties. In comparison to the existing decision procedure EAHyper, MGHyper not only handles the much larger class of hyperproperties, it also outperforms, as our experiments show, EAHyper within the $\exists^*\forall^*$ fragment.

\bibliographystyle{splncs03}
	
\bibliography{lib}

	
\end{document}